\documentclass[journal]{IEEEtran}
\usepackage{amsmath}    % need for subequations
\usepackage{graphics,epsfig,subfigure}    % need for figures
\usepackage{verbatim}   % useful for program listings
\usepackage{color}      % use if color is used in text

\usepackage{subfigure}  % use for side-by-side figures
\usepackage{hyperref}   % use for hypertext links, including those to exte$RN$al documents and URLs
\usepackage{amssymb}
\usepackage{epstopdf}
\usepackage{graphicx}
\usepackage{wasysym}
\usepackage{wrapfig}
\usepackage{sidecap}
\usepackage{float}
\usepackage{subfigure}
\usepackage{enumerate}
\usepackage{graphicx}
\usepackage{algorithmic}
\usepackage{algorithm}
\usepackage{multirow}

\allowdisplaybreaks

\newtheorem{thm}{Theorem}
\newtheorem{coro}{Corollary}
\newtheorem{lemma}{Lemma}

\begin{document}
%
% paper title
% can use linebreaks \\ within to get better formatting as desired
% Do not put math or special symbols in the title.
\title{Electric Vehicle Charging Station Placement: Formulation, Complexity, and Solutions}

\author{Albert Y.S. Lam,
        Yiu-Wing Leung, and
        Xiaowen Chu%\vspace{-9mm}% <-this % stops a space
\thanks{A preliminary version of this paper was presented in \cite{confversion}.}
\thanks{The authors are with the Department
of Computer Science, Hong Kong Baptist University, Kowloon Tong,
Hong Kong (e-mail: \{ayslam, ywleung, chxw\}@comp.hkbu.edu.hk).}% <-this % stops a space
%\thanks{Manuscript received September 19, 2005; revised December 27, 2012.}
}

% The paper headers
%\markboth{IEEE Transactions on Smart Grid,~Vol.~x, No.~x, December~201x}%
%{Lam \MakeLowercase{\textit{et al.}}: Electric Vehicle Charging Station Placement: Formulation, Complexity, and Solutions}
% The only time the second header will appear is for the odd numbered pages
% after the title page when using the twoside option.
% 
% *** Note that you probably will NOT want to include the author's ***
% *** name in the headers of peer review papers.                   ***
% You can use \ifCLASSOPTIONpeerreview for conditional compilation here if
% you desire.

% If you want to put a publisher's ID mark on the page you can do it like
% this:
%\IEEEpubid{0000--0000/00\$00.00~\copyright~2012 IEEE}
% Remember, if you use this you must call \IEEEpubidadjcol in the second
% column for its text to clear the IEEEpubid mark.

\maketitle

\begin{abstract}
To enhance environmental sustainability,  many countries will electrify their transportation systems in their future smart city plans. So the number of electric vehicles (EVs) running in a city will grow significantly. There are many ways to re-charge EVs' batteries and charging stations will be considered as the main source of energy.  The locations of charging stations are critical; they should not only be pervasive enough such that an EV anywhere can easily access a charging station within its driving range, but also widely spread so that EVs can cruise around the whole city upon being re-charged. Based on these new perspectives, we formulate the Electric Vehicle Charging Station Placement Problem (EVCSPP) in this paper. We prove that the problem is non-deterministic polynomial-time hard. We also propose four solution methods to tackle EVCSPP and evaluate their performance on various artificial and practical cases. As verified by the simulation results, the methods have their own characteristics and they are suitable for different situations depending on the requirements for solution quality, algorithmic efficiency, problem size, nature of the algorithm, and existence of system prerequisite.
\end{abstract}

% Note that keywords are not normally used for peerreview papers.
%\vspace{-2mm}

\IEEEpeerreviewmaketitle

\section*{Nomenclature}
\addcontentsline{toc}{section}{Nomenclature}
\begin{IEEEdescription}[\IEEEusemathlabelsep\IEEEsetlabelwidth{$V_1,V_2,V_3$}]
\item[$G$] The undirected graph modeling the city.
\item[$\mathcal{N}$] Set of potential charging station construction sites.
\item[$\mathcal{E}$] Set of connections connecting pairs of the construction sites.
\item[$n$] Size of $\mathcal{N}$.
\item[$d(i,j)$] Distance of the shortest path from nodes $i$ to $j$.
\item[$f_i$] Charging capacity of node $i$.
\item[$F_i$] Demand requirement of node $i$.
\item[$D$] Average traversable distance of fully charged electric vehicles.
\item[$\mathcal{N}'$] Set of nodes with charging stations constructed.
\item[$\alpha$] A discount factor.
\item[$h_{ij}$] Number of hops of the shortest path from nodes $i$ to $j$ in $G$.
\item[$x_i$] Boolean variable for construction at node $i$. 
\item[$x$] Vector of $x_i$'s.
\item[$c_i$] Construction cost at node $i$. 
\item[$\mathcal{N}_i^{\alpha D}$] Set of nodes within distance $\alpha D$ from node $i$. 
\item[$\hat{G}$] Induced graph from $G$. 
\item[$\mathcal{\hat{N}}$] Set of nodes in $\hat{G}$.
\item[$\mathcal{\hat{E}}$] Set of edges in $\hat{G}$.
\item[$H$] Induced subgraph of $\hat{G}$.
\item[$0^i$] Source node of flow attached to node $i$.
\item[$x_0^i$] Residue of flow remained in $0^i$.
\item[$y_{jk}^i$] Amount of flow on edge $(j,k)$ originated from $0^i$.
\item[$N(H)$] Set of nodes associated to $H$.
\item[$C$] A cost bound.
\item[$\tilde{G}$] Undirected graph for the vertex cover problem.
\item[$\mathcal{\tilde{N}}$] Set of nodes in $\tilde{G}$.
\item[$\mathcal{\tilde{E}}$] Set of edges in $\tilde{G}$.
\item[$\mathcal{\overline{N}}$] Set of node for node selection in the greedy algorithm.
\end{IEEEdescription}

\section{Introduction}

\IEEEPARstart{D}{ue} 
%Modern civilization relies heavily on fossil fuels to support construction, military protection, and people's mobility. Due 
to the world's shortage of fossil fuels, nations compete to secure enough reserves of natural resources for sustainability. Seeking alternative energy sources becomes crucial to a nation's future development. One of the major fossil fuel consumptions is transportation. Many daily heavily demanded vehicles are powered by gasoline. A major consequence of burning fossil fuels is the release of tremendous amount of harmful gases, which partially constitutes the global warming effect and deteriorates people's health. Electricity is considered as the most universal form of energy, which can be transformed from and to another form effectively. By converting the endurable renewable energy, like solar and wind energies, to electricity, we can manipulate energy in a much cleaner manner. Electrification of transportation, like deployment of electric vehicles (EVs), can not only alleviate our demand on fossil fuels, but also foster a better environment for living. Therefore, EVs will become the major components in the future transportation system.

EVs take the central role in this paper and they have been being studied actively since the boom of the smart grid.
Incorporating EVs into an existing self-contained transportation system is challenging. Solely expanding the population of EVs in a city without enough road connections and corresponding charging and parking infrastructure will suppress the practicability of EVs due to their limiting moving ranges. 
%Conversely, constructing the facilities with low utilization will result in a waste of resources. 
Moreover, existing gas stations are primarily designed for gas refueling; combining charging infrastructure with the conventional gas stations may not be appropriate as the relatively longer charging process will saturate the limited space of the gas stations. We need to carefully plan EV charging facilities to modernize the transportation system. To be precise, we study how EVs will be integrated into the transportation system seamlessly with a focus on charging stations and this will help make our cities ``smart''. 

We study the Electric Vehicle Charging Station Placement Problem (EVCSPP) by finding the best locations to construct charging stations in a city.
An EV should always be able to access a charging station within its capacity anywhere in the city. Charging stations should be built widely enough such that the moving range of an EV can be extended to every corner of the city by having the EV re-charged at a charging station available nearby. We study the locations where charging stations should be constructed in a city such that we can minimize the construction cost with coverage extended to the whole city and fulfillment of drivers' convenience. In this paper, we formulate the problem as an optimization problem, based on the charging station accessibility and coverage in the city. We also study its complexity and propose various methods to solve the problem.

In this paper, we focus on the long-term human aspects rather than the technological ones. 
The smart city plan and technology advancement take different time-spans for realization. 
To meet the government policy in some countries, the population of EVs needs to be boosted. Satisfaction and convenience of drivers have strong impacts on the growth of EVs in a city. The ease of re-charging their EVs is one of the most important considering factors when one decides to buy an EV \cite{consumer}. Population density and the demand for charging facilities for their EVs are passive factors. The influence of these human factors usually takes longer (say 5-10 years) to be realized. On the other hand, technology advances in a much faster pace and the impact of the charging loads to the grid will be lessened with practical technological solutions, especially for security and reliability issues.

\textcolor{black}{As a whole, the complete charging station problem with consideration of all possible considering factors can be framed as a two-level problem. In the first level, a set of potential locations for charging station constructions can be determined based on some urban planning factors, e.g., land use type, environment impact, and  safety, and also some engineering factors addressed in some of the previous work explained in the next section. In the second level, charging station placement is further enhanced from the drivers' perspective and we place charging stations in the potential locations determined from the first level. So this work mainly falls into the second level. This arrangement allows us to focus on examining the problem from a new angle.}
Furthermore, a model without too much technical details of the charging facilities allows us to retain flexibility for different charging technologies and standards. For example, charging with connected power cables can be replaced by battery swapping. Our model can still be applied to the scenario with battery swapping EVs. We focus on the human factors and it can be served as the foundation for various kinds of charging specifications.

Our main contributions include formulating the new problem EVCSPP, analyzing its complexity, and proposing several solutions to the problem.
The rest of this paper is organized as follows. Related work is given in Section \ref{sec:relatedwork}. We formulate the problem in Section \ref{sec:formulation} and discuss its complexity in Section \ref{sec:complexity}. Section \ref{sec:solutions}  presents four solution methods. In Section \ref{sec:performance}, simulation results are provided for performance evaluation and we also compare the solution methods in terms of characteristics and suitability for different situations. Finally we conclude this paper in Section \ref{sec:conclusion}.

%----------------------------------------------------------------------
%----------------------------------------------------------------------
\section{Related Work} \label{sec:relatedwork}

Most of the existing work on EVs is related to studying the operational influence of EVs on the grid, i.e., how power is transferred from and to the grid. Besides charging scheduling \cite{langtong}, in a vehicle-to-grid (V2G) system, hundreds of EVs are coordinated to act as a power source selling power back to the grid or to support auxiliary services, like regulation. A multi-layer market for V2G energy trading was proposed in \cite{doublelayer}. The market price was settled via double auction and the proposed mechanism could maximize the EVs' revenues. 
In \cite{regulation}, a queueing network was utilized to model the dynamics of EVs participating in V2G. The model could facilitate service contract engagement for regulation ancillary services.
Ref. \cite{EVUC} investigated the joint scheduling of EVs and unit commitment and this allowed us to optimize the system's total running cost with the presence of EVs. 
Ref. \cite{TVT} discussed the incorporation of PV equipment into charging stations.  It considered that charging facilities equipped with PV panels and the stored solar energy, together with the power requested from the grid, can be used to power EVs. 
%It aims to determine the optimal local energy storage size and to derive strategies swapping operations between grid-connected and islanding modes.
Refs. \cite{TPD} and \cite{IET} studied the impact of EV charging to the performance of power distribution networks with the presence of charging stations, which can represent rapid heavy loads. Ref. \cite{TPD} illustrated the effect of fast-charging EVs in terms of power-flow, short-circuit and protection while \cite{IET} proposed a new smart load management strategy to coordinate EVs for peak load shaving, power loss minimization, and voltage profile improvement. 
However, this paper is dedicated to studying the locations for building charging stations, which is an important aspect of the smart city plan.
%no longer focus on a standalone system but enlarge the scope to the city level, with particular interest in charging stations. We concern more about how EVs influence the growth of a smart city.

Both \cite{TPD2} and \cite{PSO} investigated the location and sizing issues of charging stations; \cite{TPD2} handled the two issues separately while \cite{PSO} considered a joint optimization for both. In consideration of environmental factors (e.g., load locations, load balance, power quality, etc.) and service radius of charging stations, candidate sites in \cite{TPD2} were selected with a two-step screening method instead of optimization. Ref. \cite{PSO} constructed an optimization problem in which various kinds of costs (including construction, operating, and charging costs) were minimized with traffic flow and charging requirement constraints and particle swarm optimization heuristic was adopted to compute the solution of the non-convex problem.
Ref. \cite{siting} studied the siting and sizing issues, where the locations and numbers of chargers at each site are determined at the same time with consideration of charging demand. %, but this reduces its flexibility when technologies further advances.
 Ref. \cite{sizing_PV} discussed how to allocate charging stations with the presence of solar generation.
Ref. \cite{seattle} determined the charging station locations with real-world public parking information of Seattle as inputs. It formed a mixed-integer program (MIP) by minimizing the total access costs to drivers' destinations from charging stations. 
Ref. \cite{superfast} discussed the design of power architectures and power electronics circuit topologies for high power superfast EV charging stations with enhanced grid support functionality.
Another related problem is the Gas Station Problem described in \cite{gas_station}. However, it was not related to gas station placement but determined the cheapest route connecting gas stations with other locations.
In operations research, the study of placing facilities, such as gas stations and fire stations, is generally cast as facility location problems \cite{facility_location}, e.g., the Maximal Covering Location Problem \cite{max_cover}. It concerned about the distances or times to travel to individual facilities from various locations. 
Such a model cannot guarantee that the induced subgraph constituted by the facility locations is connected but this condition is significant in our charging station placement model. 
Moreover, \cite{mayfield} gave a general discussion about the interior design of charging stations in various parking facility types instead of analyzing in the engineering perspective.
However, in this paper, we focus on the long-term issues of charging station placement for smart city planning, where the short-term factors (e.g., instantaneous loads) will be of relatively less importance. To the best of our knowledge, we are the first to study charging station placement from the new perspectives of the drivers' convenience and EVs' accessibility. Other factors, like traffic conditions, may also be taken into account but they are out of the scope of this work.

%%%%Ref. \cite{siting} also study the siting and sizing issue. The locations and  numbers of chargers at each site are determined at the same time with the consideration of  charging demand, but this reduces its flexibility when technologies further advances.

%%%%

%%%%\cite{superfast} discussed the design of power architectures and power electronics circuit topologies for high power superfast EV charging stations with enhanced grid support functionality.
%%%%
%%%However, none of them examined the placement issue in terms of EVs' reachability and accessibility for city planning.
%%%%
%%%\cite{seattle} considered public parking information in Seattle to assign EV charging station locations but it aimed to minimize the total access costs to drivers' destinations from charging stations. 
%%%
%%%%Another related problem is the Gas Station Problem described in \cite{gas_station}. However, it was not related to gas station placement but determined the cheapest route connecting gas stations between two locations.
%%%%
%%%In operations research, the study of placing facilities, such as gas stations and fire stations, is generally cast as facility location problems \cite{facility_location}, e.g., the Maximal Covering Location Problem \cite{max_cover}. It concerned for the distances or times for travel to individual facilities from locations in a graph. 

%----------------------------------------------------------------------
%----------------------------------------------------------------------
\section{Problem Formulation} \label{sec:formulation}
%We first provide the system model and then formulate the problem as an optimization program.
\subsection{System model}
We model a city with an undirected graph $G=(\mathcal{N},\mathcal{E})$, where $\mathcal{N}$ and $\mathcal{E}$ denote the sets of possible sites for constructing charging stations and connections between pairs of  sites, respectively. Suppose $|\mathcal{N}|=n. $ 
%The set $\mathcal{E}$ is associated with a function $w:\mathcal{E}\rightarrow \mathbb{R}^+$. For simplicity, for $i,j\in \mathcal{N}$, we denote $w(i,j)$ by $w_{ij}$. 
Let $d: \mathcal{N}\times \mathcal{N}\rightarrow \mathbb{R}^+$ be the distance function, where $d(i,j)$ denotes the distance of the shortest path from nodes $i$ to $j$ by traversing the connections.\footnote{The distance $d(i,j)$ refers to the distance of an actual path connecting locations $i$ and $j$ but not the Euclidean distance.} Let $f_i$ be the capacity of node $i$ representing the average capacity of charging service supported if a charging station is constructed at location $i$. It is related to the size of the site and traffic conditions in the surrounding. Each node $i$ also has a demand requirement $F_i$, which refers to its average local charging demand. The more EVs are based at location $i$, the higher $F_i$ is. $F_i$ can be estimated from the population size and the EV penetration rate of that location. Without loss of generality, we assume that some $F_i$'s are positive while some are of zero value.

We define $D$ to be the average distance able to be traversed by most typical EVs available in the market when being fully charged. A subset of nodes $\mathcal{N}'\subset \mathcal{N}$ is said to be reachable by $D$ if the following conditions hold:

 \renewcommand{\labelenumi}{C\arabic{enumi})}
\begin{enumerate}
\item For each $i\in \mathcal{N}'$, there exists a node $j\in \mathcal{N}'$ such that $d(i,j)\leq D$;
\item For each $i\in \mathcal{N}$, the total capacity, constituted from those nodes $j\in \mathcal{N}'$ such that $d(i,j)\leq \alpha D$ with discount factor $\alpha\in (0,1]$, is greater than or equal to $F_i$; and
\item For any $i,j\in \mathcal{N}'$, suppose $h_{ij}$ be the number of hops of the shortest path from $i$ to $j$ in $G$. The distance of the path $d(i,j)$ should be smaller than or equal to $h_{ij}D$. 
\end{enumerate}

$\mathcal{N}'$ represents the set of locations which have been selected with charging stations constructed. A city is well planned if $\mathcal{N}'$ is reachable by $D$. With Condition C1, an EV, which has been fully charged at one location, can re-charge at another site within distance $D$ away. C1 guarantees that EVs will not be confined in one single location (or area). Condition C2 says that the local charging demand at a location (e.g., $F_i$ at node $i$) must be satisfied by the total charging capacities contributed by those charging stations located within distance $\alpha D$ away. $\alpha$ is used to model the tolerance of drivers to move away from their current locations for re-charging. Its maximum value is one because an EV can traverse for a distance at most $D$. The smaller $\alpha$, the more conservative the model is, i.e., more charging stations should be placed around every possible location. With Condition C3, the charging station network, where each charging station is separated with another of at most distance $D$, spans the whole city. 
Note that we use one single $D$ to characterize the accessibility of the whole city for all kinds of EV models because the distribution of the charging stations should cater for all possible EVs traveling on the roads. To do this, we should assign $D$ with a more conservative value, e.g., the maximum travel distance of the most basic EV model in the market when being fully charged. 
To summarize, the conditions all together guarantee that the serving areas of the charging stations cover every corner of the city for all possible EVs.

\subsection{Formulation} \label{subsec:formualtion}
Let $x_i$ be the decision (Boolean) variable indicating if node $i$ is chosen for placement %a function $p_s:\mathcal{V}\rightarrow\mathcal{V}$, which returns a vertice $j$ with input $i$ such that the path connecting $i$ and $j$ is the shortest. In other words, there exists a path $(v_0,v_1,\ldots,v_k)$ with $v_0=i$ and $v_k=j$ and $\sum_{l=1}^k w_{v_{l-1}v_l}$ is minimum.
and $c_i$ be its construction cost. 
We minimize the total cost as the objective, i.e., $\sum_{i=1}^n c_ix_i$.
%\begin{align*}
%\sum_{i=1}^n c_ix_i.
%\end{align*} 

For each $i$, we define $\mathcal{N}_i^{\alpha D}=\{j\in \mathcal{N}|d(i,j)\leq \alpha D\}$, representing the set of nodes (including node $i$ itself) within distance $\alpha D$ from $i$. We can re-state Condition C2 as $\sum_{j\in \mathcal{N}_i^{\alpha D}}f_jx_j \geq F_i, \forall i\in\mathcal{N}$.
%\begin{align*}
%\sum_{j\in \mathcal{N}_i^{\alpha D}}f_jx_j \geq F_i, \forall i\in\mathcal{N}.
%\end{align*}
%
As Condition C3 holds for any pair of nodes, C3 implies C1. 
To re-state C3, we first create a graph $\hat{G} = (\hat{\mathcal{N}},\hat{\mathcal{E}})$, where $\hat{\mathcal{N}}$ is set to $\mathcal{N}$ and $\hat{\mathcal{E}}$ is equal to $\{(i,j)|i,j\in \mathcal{N}, d(i,j)\leq D,i\neq j\}$ (see the example shown in Fig. 1 of \cite{confversion}). 
%Consider the example of $G$, composed of 8 nodes, given in Fig. \ref{fig:G}, where the number on a connection indicates the distance between the nodes on the two ends. With $D=6$, we have the corresponding $\hat{G}$ in Fig. \ref{fig:Ghat} (ignore node $0^1$, which will be explained later). 
Consider those nodes $i$ in $G$ with $x_i=1$ (i.e., $\mathcal{N}'$) and they constitute the corresponding induced subgraph $H$ of $\hat{G}$. Condition C3 is equivalent to having $H$ connected. In other words, $H$ has one single connected component. Instead of inspecting the original graph $G$, we can focus on $\hat{G}$ to formulate the problem. 
Similar to \cite{connectedsubgraph}, we adopt a network flow model to address C3. Consider that there is some virtual flow\footnote{The virtual flow here is independent of the traffic flow.} flown from some sources to some sinks. If the sources and the sinks are not connected, the flow from the sources cannot reach the sinks. Imagine that there is a source node $0^i$ attached to node $i$ and it has $n$ units of flow available to be sent along $\hat{G}$ through node $i$. 
%For example, we attach a source node $0^i$ to node 1 in Fig. \ref{fig:Ghat} . 
Let $0\leq x_0^i\leq n$ be the residue of flow not consumed by the network. Each node $j$ with $x_j=1$ will consume one unit of flow. For each edge $(j,k)\in\mathcal{\hat{E}}$, we indicate the amount of flow on $(j,k)$ originated from $0^i$ with variable $y_{jk}^i$. Hence, we can guarantee that the flow can reach those nodes $j$ with $x_j=1$ from node $i$ on $\hat{G}$ with the following:
\begin{align}
x_0^i + y_{0i}^i &= n, \label{subcon1}\\
0\leq y_{jk}^i &\leq nx_k, \forall (j,k)\in \mathcal{\hat{E}}\cup (0^i,i),\label{subcon2}\\
\sum_{j|(j,k)\in \mathcal{\hat{E}}}y_{jk}^i &= x_k + \sum_{l|(k,l)\in \mathcal{\hat{E}}}y_{kl}^i, \forall k\in\mathcal{\hat{N}}\label{subcon3}\\
\sum_{j\in\mathcal{\hat{N}}}x_j &= y_{0i}^i,\label{subcon4}\\
0&\leq x_0^i.\label{subcon5}
\end{align}
Eq. $\eqref{subcon1}$ says that the total amount of flow $y_{0i}^i$ going out of the source $0^i$ and the retained $x_0^i$ in $0^i$ is $n$, where $n$ is the number of nodes in $G$ and it is the upper bound of flow possible to be absorbed in the network. Eq. $\eqref{subcon2}$ confines that only a sink can receive incoming flow and $\eqref{subcon3}$ describes that the total incoming flow to a node is equal to the total outgoing flow plus the amount for a sink. Eq. $\eqref{subcon4}$ explains that the total flow getting out of the source is equal to the total absorbed by the sinks and $\eqref{subcon5}$ restricts non-negative residue remained in the source.
\textcolor{black}{An illustrative example of the network flow model is given the appendix.}

%\begin{figure}[!t]
	%\begin{center}
	%\subfigure[$G$]{\label{fig:G}\includegraphics[width=3.3in]{map2.pdf}}\vspace{-0mm}
    		%\subfigure[$\hat{G}$]{\label{fig:Ghat}\includegraphics[width=3.0in]{map3.pdf}}\\
	%\end{center}
	%%\vspace{-4mm}
	%\caption{An example of 8 nodes.}
  %\label{fig:example}
%\end{figure}

Note that \eqref{subcon1}--\eqref{subcon5} require node $i$ %, a one-hop neighbor of Node $i$, or both 
to be selected for charging station construction. Otherwise, no flow from Source $0^i$ is allowed to be delivered to the sinks. To cater for this requirement, we attach a source node to each node in $\mathcal{\hat{N}}$ and the overall mathematical formulation of EVCSPP is modified accordingly as follows:
\vspace{-2mm}
\begin{subequations}
\small
\label{originalprob}
\begin{align}
\text{minimize}\quad 	& \sum_{i=1}^n{c_ix_i}  			 \label{Oobj}\\
\text{subject to}\quad %& \sum_{j\in N_i^D}x_j \geq x_i, \forall i 	\label{cond1}    \\
& \sum_{j\in \mathcal{N}_i^{\alpha D}}f_jx_j \geq F_i, \forall i	\label{OFcons}    \\
& x_i=\{0,1\}, \forall i										\label{Obooleancons}\\
%& \eqref{subcon1},\eqref{subcon2},\eqref{subcon3},\eqref{subcon4},\eqref{subcon5}, \nonumber
&  x_0^i+y_{0i}^i = n, \forall i\in \mathcal{\hat{N}} \label{Osubcon1}\\
& 0\leq y_{jk}^i \leq n x_ix_k, \forall (j,k)\in \mathcal{\hat{E}}\cup (0^i,i),\forall i\in\mathcal{\hat{N}}\\
&\sum_{j|(j,k)\in\mathcal{\hat{E}}}y_{jk}^i = x_ix_k + \sum_{l|(k,l)\in\mathcal{\hat{E}}}y_{kl}^i, \forall i,k\in\mathcal{\hat{N}}\label{Osubcon3}\\
& x_i\sum_{j\in\mathcal{\hat{N}}}x_j = y_{0i}^i,\forall i\in\mathcal{\hat{N}} \label{Osubcon4}\\
& 0\leq x_0^i,\forall i\in \mathcal{\hat{N}} \label{Osubcon5}.
\end{align}
\end{subequations}
Eqs. \eqref{Oobj} and \eqref{OFcons} have been discussed before. Eq. \eqref{Obooleancons} confines $x_i$ to be a Boolean variable. Eqs. \eqref{Osubcon1}--\eqref{Osubcon5} correspond to the induced connected subgraph condition (i.e., \eqref{subcon1}--\eqref{subcon5}) for all nodes.

Problem \eqref{originalprob} is an MIP with Boolean variables $x_i$'s and continuous variables $y_{jk}^i$'s. With the quadratic terms in Equality Constraints \eqref{Osubcon3} and \eqref{Osubcon4}, it is not a mixed-integer linear program (MILP). Hence, this problem is not easy to be solved. 

\textcolor{black}{
Before examining the complexity of the problem, we discuss the relationship of our model with the grid.
The range of power demand from a charging station can be inferred from the scale of the charging station (in terms of the number of chargers) and the usage pattern. With this information, the utility company which manages the distribution network can assess the risk of potential security problems from the expected loads of the charging stations. In the initial set of potential locations, we only consider those feasible places which allow power facility upgrade. 
We can make use of the charging capacity defined in Constraint \eqref{OFcons} to model this.
Moreover, many practical solutions for security and reliability are available to be incorporated into the grid easily. The charging stations can also be equipped with energy storage and renewable energy generation (e.g., from solar photovoltaic (PV) setup). In addition, practical methods, like installation of electric spring \cite{electricspring} and distributed active and reactive power injection control \cite{voltageregulation}, can be easily adopted to regulate the voltage with fast response time. Hence, the impact of sudden large energy demand leading to high voltage drop can be alleviated.
}

%----------------------------------------------------------------------
%----------------------------------------------------------------------
\section{Complexity Analysis} \label{sec:complexity}
The decision version of EVCSPP can be framed as follows:
Let $N(H)$ be the set of nodes associated to the induced subgraph $H$. Each node $i$ has a capacity $f_i\in \mathbb{Z}^+$ and a demand  $F_i$ and it is associated with the node set $\mathcal{N}_i^{\alpha D}$. Given an undirected graph $\hat{G}=(\mathcal{\hat{N}},\mathcal{\hat{E}})$, with the cost $c_i\in \mathbb{Z}^+, \forall i$, and a cost bound $C\in \mathbb{Z}^+$, does there exist an induced subgraph $H$ of $\hat{G}$ such that 
(i) For each $i\in \mathcal{\hat{N}}$, $\sum_{j\in \mathcal{N}_i^{\alpha D}\cap N(H)} f_j\geq F_i$; 
(ii) $H$ is connected; and
(iii) $\sum_{i\in N(H)}c_i\leq C$?
 \renewcommand{\labelenumi}{\arabic{enumi})}
%\begin{enumerate}
%%\item For each $i\in N(H)$, $|\mathcal{N}_i^{D}\cap N(H)|\geq 1$,
%%\item For each $i\not\in N(H)$, $|\mathcal{N}_i^{\alpha D}\cap N(H)|\geq 1$,
%\item For each $i\in \mathcal{\hat{N}}$, $\sum_{j\in \mathcal{N}_i^{\alpha D}\cap N(H)} f_j\geq F_i$,
%\item $H$ is connected, and
%\item $\sum_{i\in N(H)}c_i\leq C$?
%\end{enumerate}

\begin{thm} \label{thm:npc}
The decision version of EVCSPP is non-deterministic polynomial-time (NP)-complete.\footnote{In computational complexity theory, a decision problem is NP-complete if it is in the intersection of NP and NP-hard problem sets \cite{np_book}. There is no known method to solve such problem in polynomial time.}
\end{thm}

\begin{proof}
Similar to \cite{connectedsubgraph}, we construct a reduction from the vertex cover problem (VCP) to EVCSPP. In the graph $\tilde{G}=(\mathcal{\tilde{N}},\mathcal{\tilde{E}})$, a vertex cover is a subset of nodes $\mathcal{N'}\subset \mathcal{\tilde{N}}$ such that each edge $(i,j)\in\mathcal{\tilde{E}}$ has either $i$ or $j$, or both in $\mathcal{\tilde{N}}$. Without loss of generality, we assume $\mathcal{\tilde{E}}\neq \emptyset$. VCP determines if there exists a vertex cover $\mathcal{N}'$ of $\tilde{G}$ with $|\mathcal{N}'|\leq C$.

We create a graph $\hat{G}=(\hat{\mathcal{N}},\hat{\mathcal{E}})$, where $\hat{\mathcal{N}}=\tilde{\mathcal{N}}\cup \tilde{\mathcal{E}}$ and $\hat{\mathcal{E}}$ is constructed as follows. For each pair of distinct nodes $i,j\in \tilde{\mathcal{N}}$, we create an edge $(i,j)$ in $\hat{\mathcal{E}}$; for each $e=(i,j)\in \tilde{\mathcal{E}}$, we append $(i,e)$ and $(e,j)$ to $\hat{\mathcal{E}}$. For each $i\in\tilde{\mathcal{N}}$, its cost is set as $c_i=1$ and zero otherwise. For each $e\in\tilde{\mathcal{E}}$, we set $f_e=1$ and zero otherwise. We also set $\mathcal{N}_i^{\alpha D} = \mathcal{\tilde{E}}$ and $F_i = |\mathcal{\tilde{E}}|$ for all $i\in \mathcal{\hat{N}}$.

We claim that VCP on $\tilde{G}$ \textcolor{black}{has a} cost upper bound $C$ if and only if EVCSPP has a solution with cost at most $C$.
Let $\mathcal{N}'$ be a vertex cover of $\tilde{G}$ with $|\mathcal{N}'|\leq C$ and $H$ be the induced subgraph of $\hat{G}$ by nodes $\mathcal{N}'\cup \tilde{\mathcal{E}}$. It is easy to verify that $|\mathcal{N}_i^{\alpha D}\cap N(H)|=|\mathcal{\tilde{E}}|$ and thus $\sum_{j\in \mathcal{N}_i^{\alpha D}\cap N(H)} f_j=|\mathcal{\tilde{E}}| = F_i$. As $\mathcal{N}'$ is a vertex cover, each $e=(i,j)\in \mathcal{\tilde{E}}$ must have at least one of $i$ and $j$ in $\mathcal{\tilde{N}}$ and thus $H$ must contain an edge $(e,k)$ for some $k\in \mathcal{N}'$. Moreover, $\mathcal{\tilde{N}}$ forms a clique in $\hat{G}$. Hence, $H$ must be connected. Since each $e\in \mathcal{\tilde{E}}\subset \mathcal{\hat{N}}$ imposes no cost, $H$ has the same cost as $\mathcal{N}'$ in $\tilde{G}$. Therefore, EVCSPP has a solution with cost at most $C$.

Consider that an induced subgraph $H$ is a solution of EVCSPP. We set $\mathcal{N}' = N(H)\cap \mathcal{\tilde{N}}$. $H$ contains $\mathcal{\tilde{E}}$: As $f_j=1$ for $j\in\mathcal{\tilde{E}}$, for any $i\in\mathcal{\hat{N}}$, $F_i=|\mathcal{\tilde{E}}|$ guarantees $\mathcal{\tilde{E}}\subset N(H)$. Since $H$ is connected, each $i\in\mathcal{N}'$ must have an edge with an $e \in\mathcal{\tilde{E}}$ in $\hat{G}$. Moreover, $\mathcal{N}'$ has at most $C$ nodes. Hence, $\mathcal{N}'$ is a vertex cover of $\tilde{G}$ with $|\mathcal{N}'| \leq C$.
\end{proof}

With Theorem \ref{thm:npc}, we have an immediate corollary as follows:
\begin{coro}
EVCSPP is NP-hard.\footnote{An optimization problem is NP-hard if it is at least as hard as the hardest problems in the NP problem set \cite{np_book}.}
\end{coro}

%----------------------------------------------------------------------
%----------------------------------------------------------------------
\section{Proposed Solutions} \label{sec:solutions}
Since the problem is NP-hard, there is no trivial way to solve it. In this section, we propose four solution methods to tackle the problem. They possess their own pros and cons and one method may be more suitable for a particular situation than another.
%----------------------------------------------------------------------
\subsection{Method I: Iterative Mixed-Integer Linear Program}

Eqs. \eqref{subcon1}--\eqref{subcon5} can be used to guarantee that the solution subgraph constituted by all nodes $j$ with $x_j=1$ is connected as long as $x_i=1$. If we assume that node $i$ is one of locations for charging station construction, i.e., $x_i=1$, Problem \eqref{originalprob} can be reduced to
\vspace{-2mm}
\begin{subequations}
\label{probi}
\small
\begin{align}
\text{minimize}\quad 	& \sum_{k=1}^n{c_kx_k}\\
\text{subject to}\quad %& \sum_{j\in N_i^D}x_j \geq x_i, \forall i 	\label{cond1}    \\
& \sum_{j\in \mathcal{N}_k^{\alpha D}}f_jx_j \geq F_k, \forall k\in\mathcal{\hat{N}},    \label{oifeq}\\
& x_k=\{0,1\}, \forall k\in\mathcal{\hat{N}}										\label{OObooleancons}\\
%& \eqref{subcon1},\eqref{subcon2},\eqref{subcon3},\eqref{subcon4},\eqref{subcon5}, \nonumber
&  x_0^i + y_{0i}^i = n, \label{OOsubcon1}\\
& 0\leq y_{jk}^i \leq nx_k, \forall (j,k)\in \mathcal{\hat{E}}\cup (0^i,i),\\
&\sum_{j|(j,k)\in \mathcal{\hat{E}}}y_{jk}^i = x_k + \sum_{l|(k,l)\in \mathcal{\hat{E}}}y_{kl}^i, \forall k\in\mathcal{\hat{N}}\\
& \sum_{j\in\mathcal{\hat{N}}}x_j = y_{0i}^i,\\
& 0\leq x_0^i,\\
& x_i = 1. \label{oi1eq}
\end{align}
\end{subequations}
Problem \eqref{probi} is an MILP and it can be solved with standard MIP solvers applying methods like branch-and-bound. Now the question becomes which node $i$ should be chosen for this purpose. 

We write the solution of Problem \eqref{probi} as $\inf_{x\in\Omega_i}\sum_{j=1}^n{c_jx_j}$, where $\Omega_i$ is the solution space constituted by \eqref{oifeq}--\eqref{oi1eq} and $\inf$ refers to the infimum operator. We have the following theorem:
\begin{thm}
The solution of Problem \eqref{originalprob} can be determined by solving
\vspace{-3mm}
{\small
\begin{align}
\min_{1\leq i\leq n}(\inf_{x\in\Omega_i}\sum_{j=1}^n{c_jx_j}). \label{m1formulation}
\end{align}
}
\end{thm}
\begin{proof}
The solution of \eqref{originalprob} can induce a connected subgraph of $\hat{G}$ where \eqref{OFcons} is satisfied and the total cost function \eqref{Oobj} is minimized. When we solve \eqref{probi}, its solution induces a connected subgraph of $\hat{G}$ with minimum total cost satisfying \eqref{OFcons} \textit{where node $i$ is included}.
To compute \eqref{m1formulation}, we apply \eqref{probi} to every node $i$ (thus we solve \eqref{probi} $n$ times) and its solution is the minimum among the $n$ subgraphs. Since \eqref{Oobj} and \eqref{OFcons} exist in every \eqref{probi}, the solution of \eqref{originalprob} must be one of the computed subgraphs with some node $i$ included. Hence,  addressing \eqref{m1formulation} can solve \eqref{originalprob}.
\end{proof}

With this result, the original mixed-integer \textit{non-linear} program becomes $n$ solvable MILPs.
%As there is no trivial way to choose such a node $i$, we need to apply \eqref{probi} to every possible node. 
%In other words, we solve \eqref{probi} $n$ times, each of which has $i$ set to one of $\{1,\ldots n\}$. The solution of \eqref{originalprob} is the best feasible one of the MILPs' solutions.
%\footnote{The MILPs for some fixed nodes may result in infeasible solutions to the original problem. The best solution is selected among the feasible ones.} 
Since we need to go through all nodes iteratively, we call this method Iterative MILP.

Besides the fact that the computational time of solving MILP \eqref{probi} grows super-linearly with $n$, this method also suffers from the problem that the number of MILPs (i.e., \eqref{probi}) needed to be solved also increases with $n$. Hence, the combined effect of increasing $n$ will make its computation time accelerating extraordinarily fast. Hence, this method is only applicable to small problem instances.
If the solver applied to \eqref{probi} can produce the optimal solution, Method I will guarantee the optimality.

%----------------------------------------------------------------------
\subsection{Method II: Greedy Approach}

Here we present an efficient greedy algorithm, which is applicable to the original formulation \eqref{originalprob} and requires much shorter computation time. Before discussing its details, we have the following lemma to facilitate its development.  

\begin{lemma}\label{lemma:upperbound}
Problem \eqref{originalprob} (and Problem \eqref{probi}) is feasible if and only if $x=[x_1,\ldots,x_n] = [1,\ldots,1]$ is a feasible solution, which gives an upper bound of the objective function value $\sum_{i=1}^{n}c_i$.
\end{lemma}

\begin{proof}
First we consider the only if-direction. As the problem is feasible, there exists a feasible $x'=[x'_1,\ldots,x'_n]$, composed of some 0's and/or 1's, satisfying Constraints \eqref{OFcons}--\eqref{Osubcon5}. If $x'_i=1$ for all $i$, then we have the result. Consider that there is at least one $j$ such that $x'_j=0$. If we produce another $x''$ by modifying $x'_j$ with value one, besides \eqref{Obooleancons}, $x''$ will always satisfy Constraint \eqref{OFcons}, as we will not change or increase the sum on the left-hand side of \eqref{OFcons}. Moreover, as $0< \alpha\leq 1$, if $x'_j=0$ satisfies \eqref{OFcons}, there exists at least one node $k$ with $x_k'=x_k''=1$ within distance $D$ away from node $j$. In this way, if we have $x''_j=1$, we will attach node $j$ to the subgraph induced by $x'$ through node $k$. In other words, the subgraph induced by $x''$ is still connected, i.e., satisfying \eqref{Osubcon1}--\eqref{Osubcon5}. We can repeat this process until we change all 0's to 1's and this produces $x$ with upper bound $\sum_{i=1}^{n}c_i$.

The if-direction is trivial. We complete the proof.
\end{proof}

\begin{coro}\label{coro:infeasible}
If $x=[x_1,\ldots,x_n] = [1,\ldots,1]$ is not feasible, EVCSPP is infeasible.
\end{coro}
Corollary \ref{coro:infeasible} can be used to check the feasibility of a problem instance.

%\vspace{-3mm}
\begin{algorithm}
\caption{Greedy Algorithm} \label{greedy}
\footnotesize
\begin{algorithmic}[1]
\STATE Set $x_i=1$ for $i=1,\ldots,n$.
\REPEAT
	\STATE Construct a node set $\overline{\mathcal{N}}$ composed of nodes $i$ with $x_i=1$ where the induced subgraph is still connected when $x_i$ is set to zero.
%	\STATE Arrange those $i$'s in $\overline{\mathcal{N}}$ into a list %$\mathcal{L}=[l_1,\ldots,l_{|\overline{\mathcal{N}}|}]$ in the descending order according to $c(i)$.
%	\STATE $j \gets 1$
	\STATE $x'\gets x$
	\STATE $flag \gets 0$
	\REPEAT
		\STATE Select $j$ with the largest $c_j$ in $\overline{\mathcal{N}}$.
		\STATE Modify $x'$ by setting $x'_{j}=0$
		\IF {$x'$ satisfies \eqref{OFcons}}
			\STATE $x\gets x'$
			\STATE $flag \gets 1$
		\ELSE
			\STATE $x'_{j}\gets 1$
			\STATE Remove $j$ from $\overline{\mathcal{N}}$
%			\STATE $j \gets j+1$
		\ENDIF
	\UNTIL {$flag = 1$} OR {$\overline{\mathcal{N}}=\emptyset$}
\UNTIL {$\overline{\mathcal{N}}=\emptyset$}
\end{algorithmic}
\end{algorithm}
%\vspace{-3mm}

Assume that we have a feasible problem instance. 
We construct a greedy algorithm by reducing the total cost as much as possible in each iteration and it results in a sub-optimal solution.
Its pseudocode is given in Algorithm \ref{greedy}. In Line 1, we start with the feasible $x=[x_1,\ldots,x_n] = [1,\ldots,1]$ explained in Lemma \ref{lemma:upperbound} and then go through a certain number of iterations (Lines 2--17). In each iteration, we select those nodes in the subgraph induced by $x$ which will not disconnect the subgraph if we remove them from the subgraph and we call this selection $\overline{\mathcal{N}}$ (Line 3). For example, Fig. \ref{fig:greedy} shows a graph $\hat{G}$ of six nodes, where a dot $i$ and a hole $j$ mean $x_i=1$ and $x_j=0$, respectively. In this case, we have $\overline{\mathcal{N}}=\{1,3,6\}$. We can see that the resultant $x'$ formed by removing any one node in $\overline{\mathcal{N}}$ will still satisfy Constraints \eqref{Osubcon1}--\eqref{Osubcon5}. Then we attempt to deselect the one (say node $j$) with the highest cost $c_j$ in $\overline{\mathcal{N}}$ (Line 7). If the resultant $x'$ satisfies \eqref{OFcons}, $x'$ is a feasible point and we proceed to the next iteration (Lines 9--11). Otherwise, we remove $j$ from $\overline{\mathcal{N}}$ (Line 14). Instead of deselecting $j$ (Line 13), we deselect the one with the next highest cost. The iterations terminate when no nodes remain in $\overline{\mathcal{N}}$ (Line 17). The final solution $x$ is the best determined by the greedy algorithm. Note that the resultant solution is usually sub-optimal, especially when the problem size $n$ becomes larger.

\begin{figure}[!t]
\centering
\vspace{-3mm}
\includegraphics[width=2.8cm]{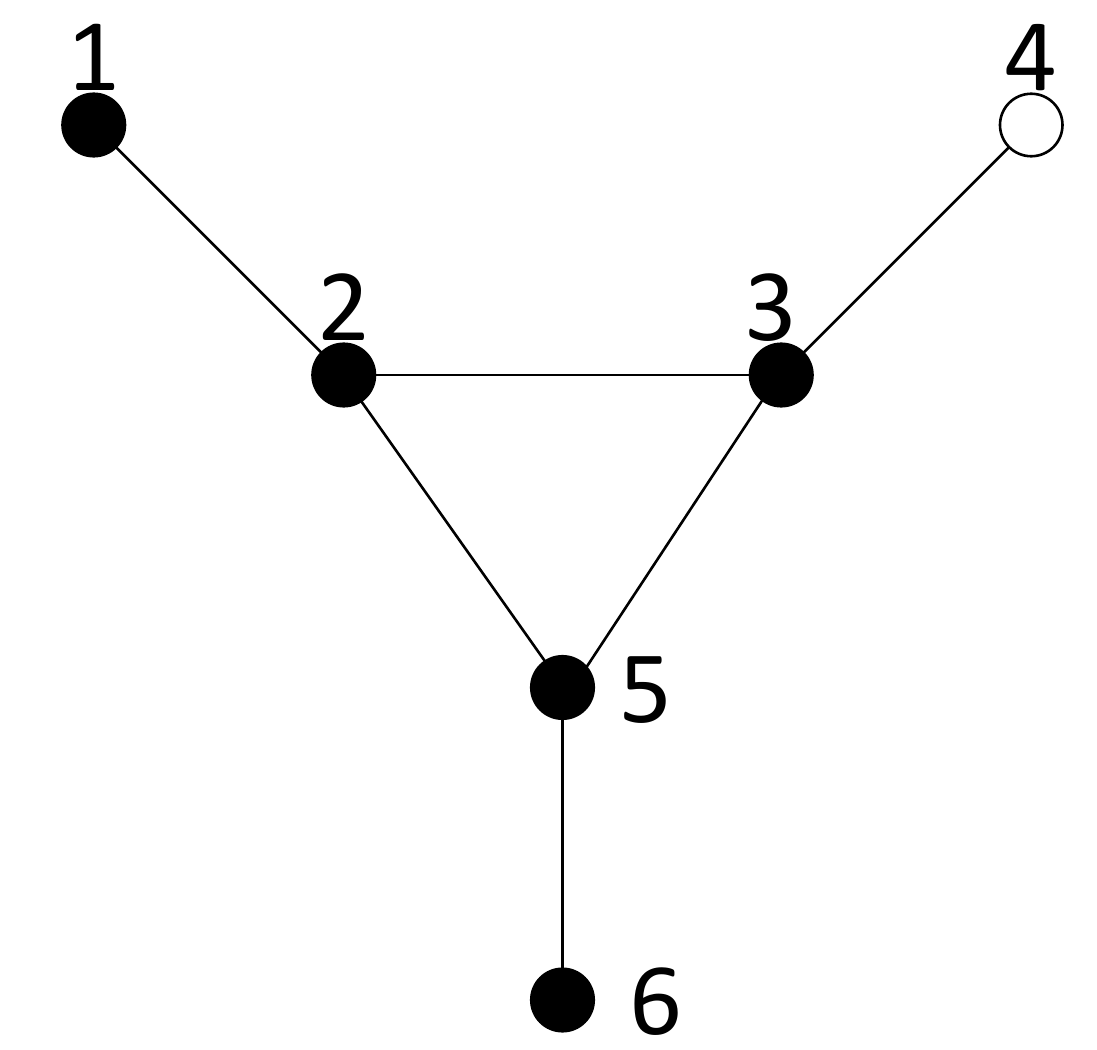}
\vspace{-3mm}
 \caption{Node selection in the greedy algorithm.}
\label{fig:greedy}
\vspace{-6mm}
\end{figure}

%----------------------------------------------------------------------
\vspace{-5mm}
\subsection{Method III: Effective Mixed-Integer Linear Program}

Recall that in Method I, we need to apply MILP \eqref{probi} to every node since we are not sure which node $i$ has $x_i=1$ in the optimal solution and thus it is usually subject to long computation time.
%Besides the fact that the computational time of solving MILP \eqref{probi} grows super-linearly with $n$, Method 1 is also suffered from the problem that the number of MILPs (i.e., \eqref{probi}) needed to be solved also increases with $n$. Hence, the combined effect of increasing $n$ will make the computational time for Method 1 accelerating extraordinarily fast. 
However, if we know the node which has unity in the solution, we can save lots of effort by applying \eqref{probi} to that node only. With Theorem \ref{thm:method3}, under a general condition, we find that not all nodes $i$ are required to generate the solution as in Method I.

\begin{thm} \label{thm:method3}
%For any node $i$, let $\mathcal{N}_i$ be the set of $i$'s one-hop neighbor nodes in $\hat{G}$ together with node $i$ itself. 
Suppose that the demand requirements $F_i$ for all $i$ are positive. Then for any node $i$, at least one node $j$ in $i$'s one-hop neighborhood in $\hat{G}$, i.e., $\mathcal{N}_i^{\alpha D}$, must have $x_j=1$ in the optimal solution of \eqref{originalprob}.  
\end{thm}
\begin{proof}
Since all $F_i$'s are positive, $[x_1,\ldots,x_n] = [0,\ldots,0]$ can never be a solution. Hence, the solution must contain at least one node $i$ with $x_i=1$.
Consider any particular node $i$. If $x_i$ is set to one in the optimal solution, then the result holds by assigning $j$ to $i$ as $i\in\mathcal{N}_i^{\alpha D}$.

Consider that $x_i$ is set to zero. As $F_i$ is positive, at least one term $f_j x_j$ on the left-hand side of \eqref{OFcons} must be positive. In other words, at least one node (e.g., $j$) in $\mathcal{N}_i^{\alpha D}$ has $x_j=1$. Since $\alpha\leq 1$, node $j$ must be a one-hop neighbor of node $i$. Since all $F_i$'s are positive, the condition applies to every node. 

Hence, the result is true for every node in the graph. 
\end{proof}

 With this theorem, we compose Method III by choosing any node $i$ and applying Problem \eqref{probi} with respect to those nodes in $\mathcal{N}_i^{\alpha D}$ only. The number of \eqref{probi} required to be solved depends on the cardinality of $\mathcal{N}_i^{\alpha D}$. We can minimize the computational time by choosing the node $i$ with the smallest degree in $\hat{G}$. In this way, we can simplify Method I by exploiting the network structure of the graph and the solutions of both Methods I and III are equivalent. Similar to Method I, If the solver applied to \eqref{probi} can produce the optimal solution, Method III will also guarantee the optimality.
Since EVs are movable in a city, it is common to have EVs appearing in every location (node) in a certain time-span, and thus we have positive $F_i$ for all nodes $i$. Hence, the condition imposed in Theorem \ref{thm:method3} generally holds in most situations.

%----------------------------------------------------------------------
\subsection{Method IV: Chemical Reaction Optimization}
Chemical Reaction Optimization (CRO) is a recently proposed Nature-inspired metaheuristic for optimization \cite{CRO}. Under certain conditions, it has been proved to be able to converge to the global optimum for combinatorial optimization problems (like EVCSPP) \cite{CRO_proof} and it has been demonstrated to have very good performance in solving real-world problems, 
e.g., \cite{CRO_grid,CRO_cognitive}.
%e.g., \cite{CRO_grid,CRO_cognitive,CRO_PMU}. 
CRO is general-purpose and we apply CRO to EVCSPP.
In CRO, the manipulated agents are molecules, each of which carries a solution. The molecules explore the solution space of the problem through a random series of elementary reactions taking place in a container. We define four types of elementary reactions, each of which has its own way to modify the solutions carried by the involved molecules. Due to space limitation, we do not illustrate every detail of CRO but explain the necessary modifications based on the framework described in \cite{CRO}. We basically follow \cite{CRO} to construct the algorithm. It consists of four elementary reactions, including on-wall ineffective collision, decomposition, inter-molecular ineffective collision, and synthesis. They are implemented as follows:

\subsubsection{On-wall ineffective collision}
It mimics that a molecule hits a wall of the container and then bounces back. This elementary reaction is not vigorous and we only \textcolor{black}{have} small modifications to the molecule. Let $x$ and $x'$ be the solutions held by the molecules before and after the change. We apply our greedy approach (Method II)\footnote{Note that we can initiate the greedy algorithm with any $x$ instead of the unity vector $[1,\ldots,1]$ by skipping Line 1 in Algorithm \ref{greedy}.} to $x$ to produce $x'$, i.e., $x \xrightarrow{\text{greedy}} x'$.
%\begin{align*}
%x \xrightarrow{\text{greedy}} x'.
%\end{align*}

\subsubsection{Decomposition}
It describes that one molecule hits a wall and breaks into two separate molecules. It involves vigorous changes to the molecules. Let $x$ be the solution held by the reactant molecule and $x_1'$ and $x_2'$ be the solutions of the resultant molecules. Here $x_1'$ and $x_2'$ are randomly generated in the solution space. A random solution can be produced by modifying Algorithm \ref{greedy} where the repeat loop (Lines 6--16) iterates for a random number of times (between $1$ to $n$) and we select a random node in $\overline{\mathcal{N}}$ in Line 7.\footnote{Although the generations of $x_1'$ and $x_2'$ do not rely on $x$, the energies stored in the molecules do. Interested readers may refer to \cite{CRO} for more information.} A decomposition can be described as $x \xrightarrow{\text{random}} x_1'+x_2'$.
%\begin{align*}
%x \xrightarrow{\text{random}} x_1'+x_2'.
%\end{align*}

\subsubsection{Inter-molecular ineffective collision}
It portrays that two molecules collide with each other and then bounce away. The change is not vigorous. Let $x_1$ and $x_2$ be the two reactant molecules and $x_1'$ and $x_2'$ be the resultant molecules. Similar to the on-wall ineffective collision, we apply the greedy algorithm to the respective molecules to modify the solutions, i.e., $x_1 + x_2 \xrightarrow{\text{greedy}} x_1' + x_2'$.
 %\begin{align*}
%x_1 + x_2 \xrightarrow{\text{greedy}} x_1' + x_2'.
%\end{align*}

\subsubsection{Synthesis}
Synthesis describes that two molecules collide with each other and then combine into one molecule. The change is vigorous. Similar to decomposition, we produce a new molecule by randomly generating its solution in the solution space. Let $x_1$ and $x_2$ be the two reactant molecules and $x'$ be the resultant molecule. We have $x_1 + x_2 \xrightarrow{\text{random}} x'$.
 %\begin{align*}
%x_1 + x_2 \xrightarrow{\text{random}} x'.
%\end{align*}

When initializing the algorithm, we assign random solutions in the solution space to the molecules (this can be done by the random solution generation used in decomposition). It is clear that Method II is embedded in this method except that Method II always start with a unity vector $x$. We can guarantee that Method IV is always superior to Method II in terms of solution quality by having at least one molecule possessing a unity vector as its initial solution. So we can assign the unity vector to some initial molecules (say 10\%) in the initialization phase. Since CRO is a probabilistic algorithm, the solutions produced in different runs could be different.

%----------------------------------------------------------------------
%----------------------------------------------------------------------
\section{Performance Study} \label{sec:performance}
\subsection{Simulation Results}

\begin{table*}[!t]
\renewcommand{\arraystretch}{1.1}
\caption{Simulation results for $n=50$ and $D=20$ km}
\label{tab:performance}
\centering
%\begin{tabular}{p{1cm}|p{1.1cm} | p{1.1cm}|p{1cm}|p{1cm}}
\begin{tabular}{c|p{1cm} || p{1.2cm}|p{1.2cm}|p{1.2cm}|p{1.2cm}|p{1.2cm}|p{1.2cm}||c|c|c|p{1.2cm}}
\hline\hline
\multirow{2}{*}{$\alpha$} & \multirow{2}{1cm}{\tiny Matched/ Feasible cases} & \multicolumn{6}{|c||}{Objective function value (No. of stations constructed)} & \multicolumn{4}{c}{Computation time (s)}\\ \cline{3-12}
 	&  	& Method I &  Method II & Method III	& Method IV (mean) & Method IV (best) & Method IV (worst)	& Method I &  Method II & Method III	& Method IV (mean) \\
\hline
1		& 32/100	&	9.4875 (24.7500)	& 9.7553 (25.7500)	& 9.4875 (24.7500)	& 9.6119 (25.2150)	& 9.5092 (24.8800)	& 9.7158 (25.5700)	& 478.6861	& 0.0610	& 23.6505	& 20.4516\\
0.9	& 20/62		&	10.8950 (27.2097)	& 11.0986 (28.0000)	& 10.8950 (27.2097)	& 11.0050 (27.5435)	& 10.9235 (27.2903)	& 11.0745 (27.8387)	& 369.5348	& 0.0618	& 18.2473	& 23.3749\\
0.8	& 8/23		&	12.5701 (30.4783)	& 12.8035 (31.2174) & 12.5701 (30.4783)	& 12.6786 (30.6826)	& 12.5818 (30.4783)	&	12.7829 (30.9130)	&	310.1252	& 0.0611	& 15.7985	& 29.6049\\
0.7 & 0/2			&	13.3116 (33.5000)	& 13.8089 (34.5000)	& 13.3116 (33.5000)	& 13.4800 (33.5000)	&	13.3182 (33.5000)	& 13.7777 (33.5000)	& 331.2596	& 0.0542	& 19.9461	& 34.9514\\
\hline\hline
\end{tabular}
%\vspace{-3mm}
\end{table*}

 We perform a series of simulations to evaluate the performance of the four solution methods.  All simulations are run on the same computer with Intel Core i7-3770 CPU at 3.40GHz and 16GB of RAM, and conducted in the MATLAB environment. For Methods I and III, the MILP is computed with the CPLEX solver \cite{cplex} and YALMIP \cite{yalmip}.  Recall that Methods I, II, and III are deterministic while Method IV is probabilistic. For illustrative purposes, we repeat Method IV 10 times for each simulation case. After several trial runs, we set the parameter values for Method IV as: ``function evaluation (FE) limit'' = 2000, ``initial kinetic energy (KE)'' = 10, ``initial population size'' = 40, ``initial buffer'' = 10, ``collision ratio'' = 0.5, ``synthesis threshold'' = 0.5, ``decomposition threshold'' = 20, and ``KE loss rate'' = 0.9.
We conduct three tests. In the first test, we examine the solutions' performance with changing $\alpha$. The second aims to study how the computation time grows with the problem size. In the third, we test how the solution methods perform in a real-world scenario.
%given in Table \ref{tab:parameters}.

%\begin{table}[!t]
%\renewcommand{\arraystretch}{1.3}
%\caption{Parameter settings for CRO (Method 4)}
%\label{tab:parameters}
%\centering
%\begin{tabular}{c|c }
%\hline\hline
%Parameter & Value	\\
%\hline
%Function evaluation limit	& 2000\\
%Initial KE				& 10\\
%Initial population size	& 40\\
%Initial buffer			& 10\\
%Collision ratio			& 0.5\\
%Synthesis threshold		& 0.5\\
%Decomposition threshold	& 20\\
%KE loss rate			& 0.9\\
%\hline\hline
%\end{tabular}
%\end{table}

In the first test, we randomly generate 100 feasible instances.
Each instance of $G$ is constructed by randomly placing 50 nodes in an area of $100\times 100$ km$^2$, where we assign a random value in the range of $(0,1]$ to the cost $c_i$, and $D$, $f_i$ and $F_i$, for all $i$, are set to 20 km, 0.5, and 1, respectively. For simplicity, we assume that the nodes are interconnected and the length of the shortest path of each pair of nodes is determined with the Euclidean distance between them. As explained in Section \ref{sec:formulation}, we can produce $\hat{G}$ from $G$. Then we can check the feasibility of each instance with Corollary \ref{coro:infeasible}.

We verify the performance of the four methods with respect to  the computed solution quality and the computation time. 
The results are given in Table \ref{tab:performance}.  
The second column indicates the number of feasible and matched cases among the 100 graphs. All graphs are feasible when $\alpha$ is equal to one. When $\alpha$ decreases, the number of resulted feasible cases will also decrease as Constraint \eqref{OFcons} becomes stronger. 
The matched cases indicate those of the feasible ones producing the same objective function values by all the four approaches. Regardless of the non-statistically significant cases with $\alpha=0.7$, all the four methods can produce the best solutions for around $\frac{1}{3}$ of the feasible cases. 
The other columns show the average objective function values and computation times of the four methods for the feasible cases. For Method IV, we also provide the average (among the cases) of the best (among the repeats) and the worst (among the repeats) for reference. 
The average numbers of charging stations appeared in the solutions are put in brackets.
Methods I and III always give the best solutions and Method IV always outperforms Method II. In terms of computation time, Method II is the fastest and Method III comes the second. Method IV is the next and Method I takes the longest.

In the second test, we study how the computation time changes with the problem size. The setting is similar but we fix $\alpha$ to one for different values of $n$. We generate 10 feasible cases for $n$ equal to 10, 50, 100, 150, and 200. Fig. \ref{fig:sim_n} shows the average computation times of the four methods in the logarithmic scale. The corresponding objective function values normalized by computed minimums are also given for reference. All the computation times increase with $n$. Method I takes the longest computation time which also grows the fastest. Method III needs less time than Method IV when $n$ is small. However, when $n$ is larger than 100, Method III requires more time to compute the solution. In other words, 
the computation time of Method III grows faster than that of Method IV. Method II needs the shortest time but its computed solution quality is the worst. Although Methods I and III require relatively more time, their solutions are the best. 
Note that the results of Methods I and III for $n=200$ are not shown because they are not computable by YALMIP/CPLEX due to the out-of-memory problem. This implies the MILP approaches are not suitable for large problems.\footnote{We run the simulations in MATLAB. The out-of-memory problem may happen at another $n$ with a different combination of machine and platform. Here we just demonstrate that Methods I and III are not scalable.}
As before, Method IV always produces better solutions than Method II. 
The average numbers of charging stations constructed by the four methods are given in Table \ref{tab:numstation}. 
Moreover, we perform a series of tests to check the convergence of Method IV. We run CRO for some cases of different sizes used in the second test with duration up to 10 000 FEs. Fig. \ref{fig:convergence} gives the convergence curves  for particular cases; for clear illustration (of the performance at the beginning), we give the performance in first 500 function evaluations only.  The results show that the algorithm converges very fast and can converge within 500 FEs in all the cases. Hence it is concluded that our evaluation limit of 2000 for Method IV applied to all the three tests is sufficient.

\begin{figure}
\centering
\includegraphics[width=3.5in]{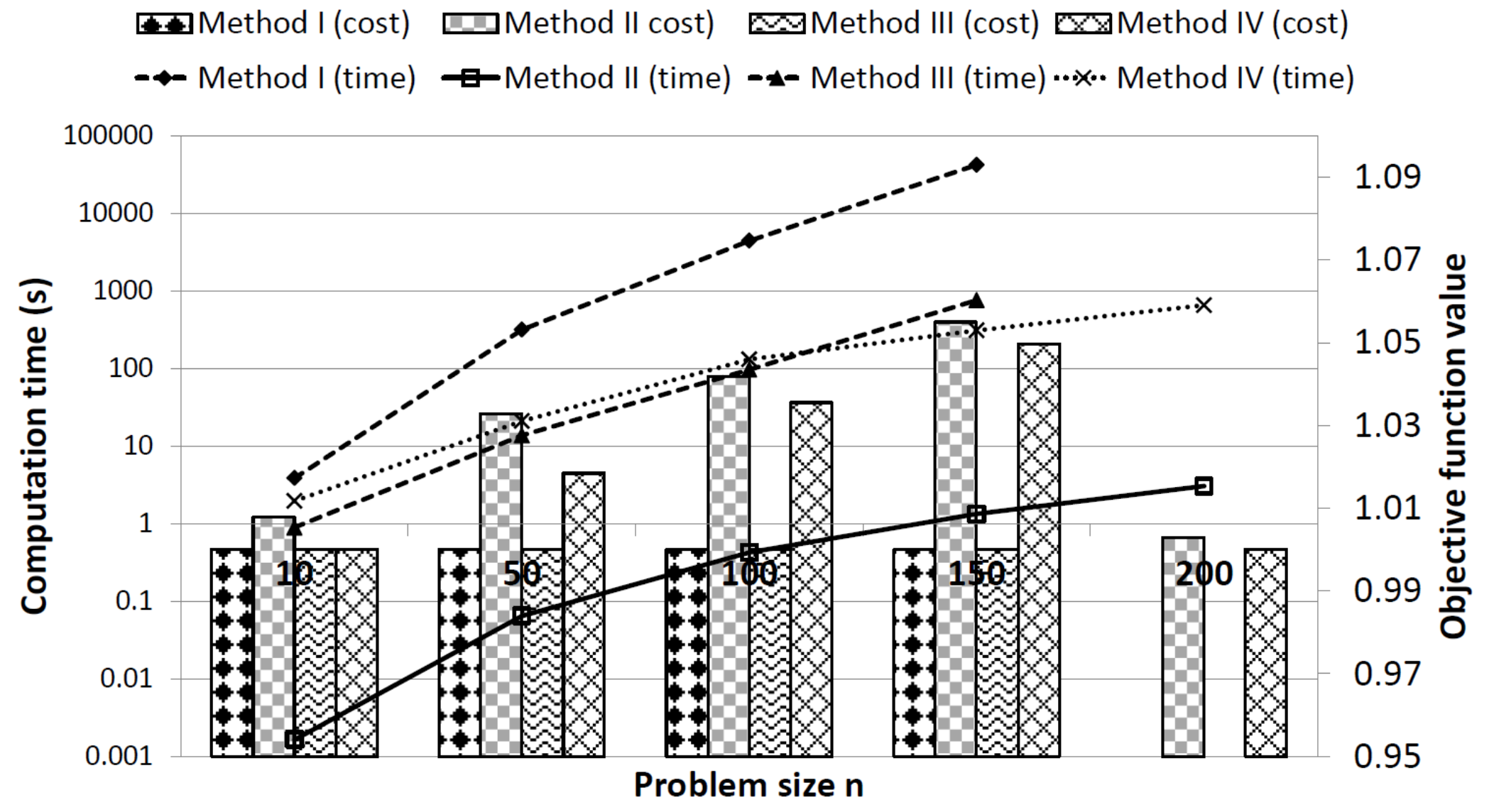} %\vspace{-5mm}
 \caption{Computation time changing with problem size.}
\label{fig:sim_n}
%\vspace{-5mm}
\end{figure}

\begin{table}[!t]
\renewcommand{\arraystretch}{1.1}
\caption{Average number of stations computed in the second test.}
\label{tab:numstation}
\centering
\begin{tabular}{c|c | c|c|c}
\hline\hline
 	& Method I 	& Method II &  Method III & Method IV\\
\hline
10		& 6.30			&	6.40		& 6.30	& 6.30	\\
50		& 24.40			& 25.60		& 24.40	& 24.90		\\
100		& 43.60			&	45.20	& 43.60 	& 44.89	\\
150		& 45.80			&	49.30	& 45.80 	& 48.45	\\
200	  & -			&	51.1	& -	& 50.54	\\
\hline\hline
\end{tabular}
\end{table}

\begin{figure}
\centering
\includegraphics[width=3.5in]{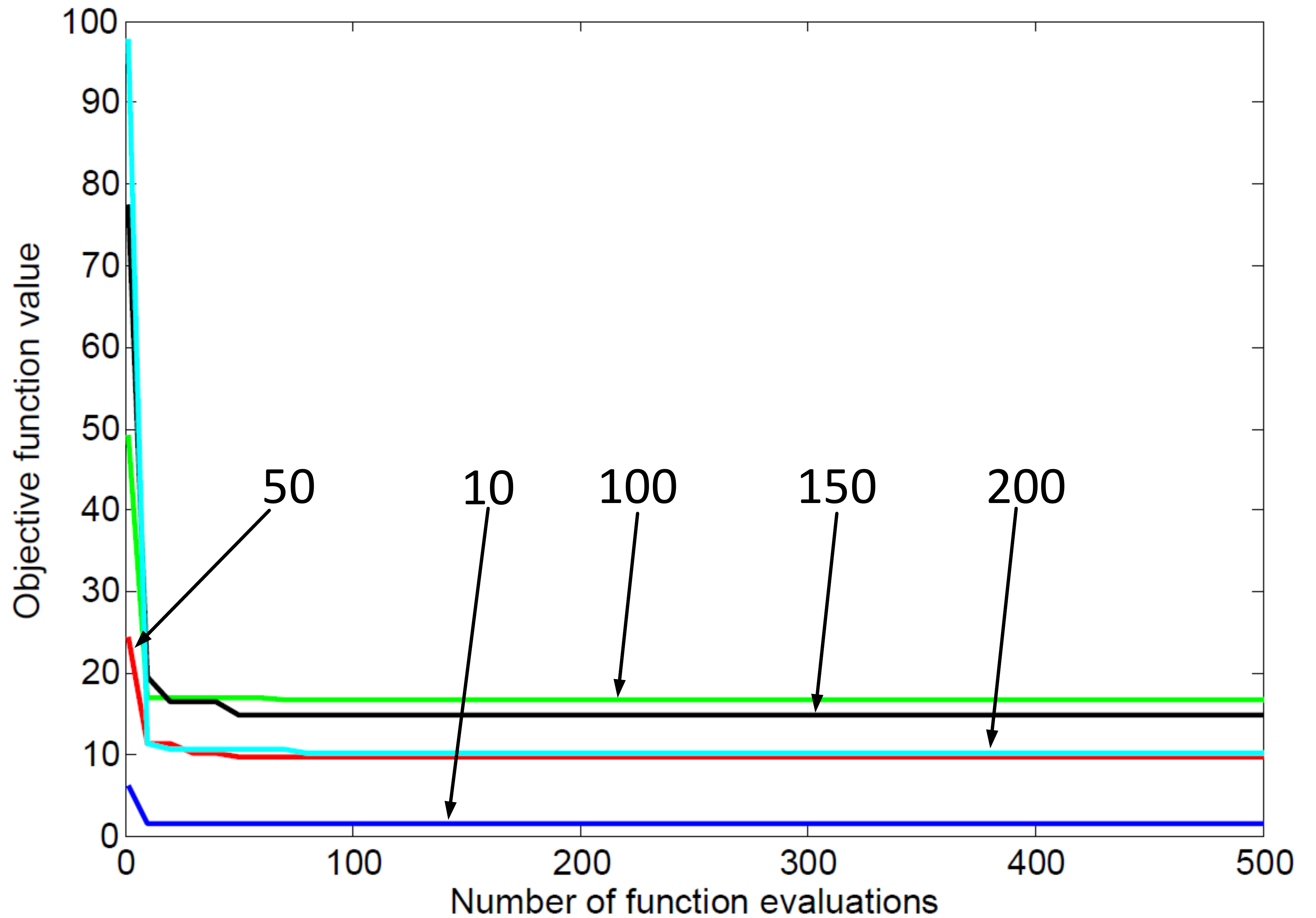} %\vspace{-5mm}
 \caption{Convergence of Method IV.}
\label{fig:convergence}
%\vspace{-5mm}
\end{figure}

In the third test, we apply the problem to a real-world environment; we determine the locations for building charging stations in Hong Kong (HK). The HK Government plans to introduce more EVs (e.g., taxi) into the city and the construction of charging stations is one of the crucial steps in the plan. We can see how this can be realized through solving EVCSPP. HK is composed of three zones (New Territories, Kowloon, and HK Island) and each zone is further divided into districts. There are total 18 districts in HK. We select one location in each district for potential charging station construction. The location distribution is given in Fig.~\ref{fig:HKmap}. The distance between each pair of locations is retrieved from the route connecting them suggested by Google Map \cite{googlemap}. We relate the location parameters to the district data obtained from \cite{18districts}. We assign the population size to the demand $F_i$ and the median monthly income per capita to the cost $c_i$. We set the capacity $f_i$ inversely proportional to the density with some proportionality constants. The location parameter values are listed in Table~\ref{tab:data}. We perform simulations with several combinations of $D$ (30, 35, 40, 45, and 50 km) and $\alpha$ (1, 0.8, and 0.6) and the performance of the methods is given in Table~\ref{tab:HKsim} where the best  objective function values are bold and the numbers of stations constructed are put in brackets. The number of nodes matched in the solutions of the four methods (the best solutions for Method IV) for each feasible case is also provided.
In Fig. \ref{fig:HKmap}, for the case with $\alpha$ and $D$ equal to $0.6$ and $45$ km, respectively, those locations in red indicate that charging stations should be constructed according to the solutions computed by the four methods; the roman numbers beside a location reveal which methods have included that location in their solutions.\footnote{As the solutions of Method IV computed in different runs can be different, only its solution of a particular run is given in Fig. \ref{fig:HKmap}.}

\begin{figure}
\centering
\includegraphics[width=3.5in]{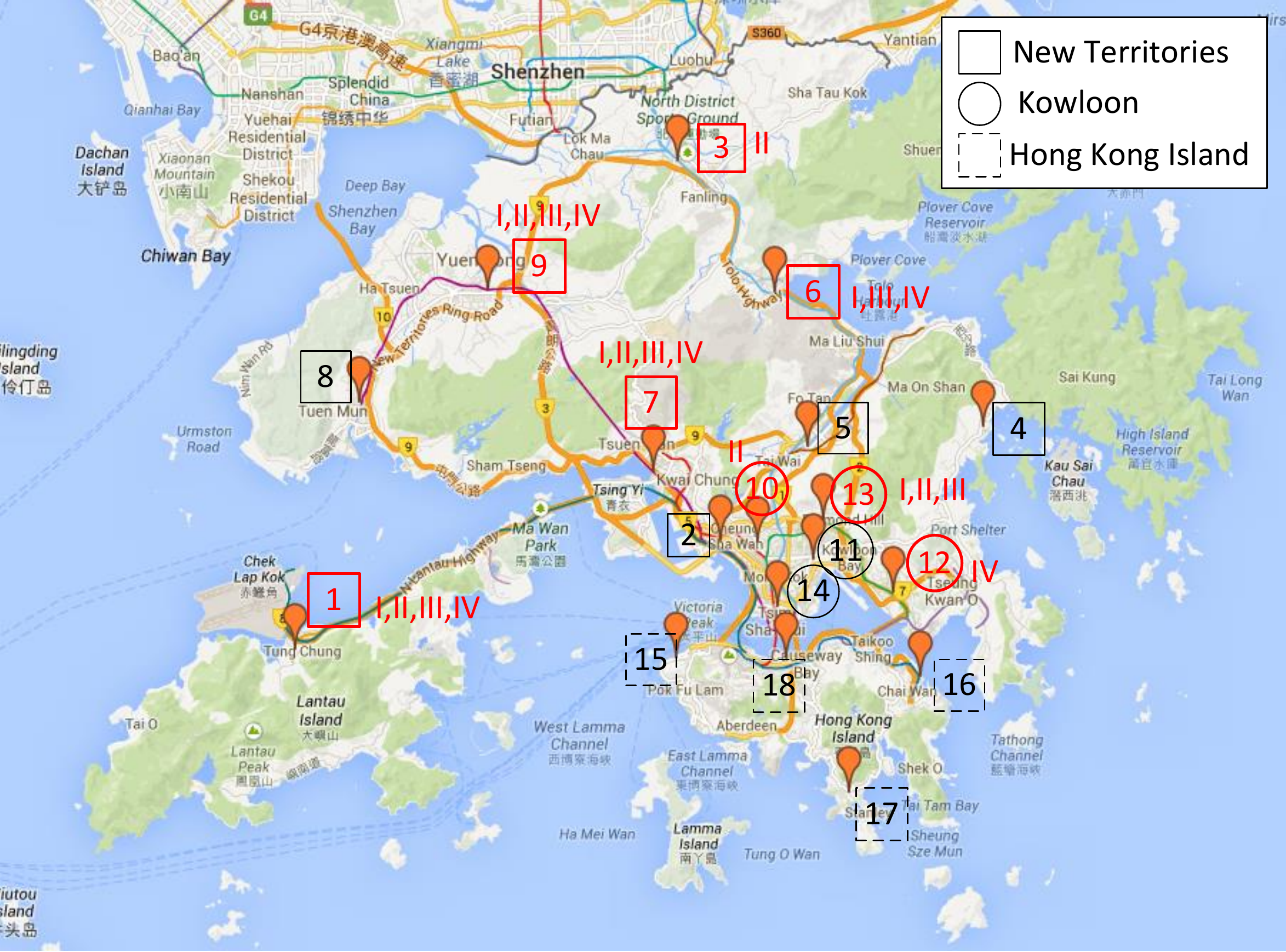}
\vspace{-3mm}
 \caption{Charging station distribution (adopted from \cite{googlemap}).}
\label{fig:HKmap}
\end{figure}

\begin{table}[!t]
\renewcommand{\arraystretch}{1.1}
\caption{Charging station data (from the 2006 census)}
\label{tab:data}
\centering
%\begin{tabular}{p{1cm}|p{1.1cm} | p{1.1cm}|p{1cm}|p{1cm}}
\begin{tabular}{c|c | c|c|c}
\hline\hline
District 	& Population (k) 	& Density (/km$^{2}$) &  \multicolumn{1}{p{2cm}|}{Median monthly per capita (HK$\$$)} & $f_i$\\
\hline
1		& 137.1			&	783		& 5659	& 1277.14	\\
2		& 523.3			&	22421	& 4833 	& 44.60	\\
3		& 280.7			&	2055	& 5161 	& 486.62	\\
4 	& 406.4			&	3135	& 6774	& 318.98	\\
5	  & 607.5			&	8842	& 6232	& 113.10	\\
6	  & 293.5			&	2156	& 5806	& 463.82	\\
7   & 288.7			&	4679	& 6897	& 213.72	\\
8   & 502.0			&	6057	& 5172	& 165.10	\\
9   & 534.2			&	3858	& 4777	& 259.20	\\
\hline
10  & 365.5			&	39095	& 4821	& 255.79	\\
11  & 362.5			&	36178	& 6897	& 276.41	\\
12  & 587.4			&	52123	& 4845	& 191.85	\\
13  & 423.5			&	45540	& 4750	& 219.59	\\
14  & 280.5			&	40136	& 6034	& 249.15	\\
\hline
15  & 250.0			&	20102	& 9722	& 99.49	\\
16  & 587.7			&	31664	& 7235	& 63.16	\\
17  & 275.2			&	7083	& 6563	& 282.37	\\
18  & 155.2			&	15788	& 10185	& 126.68	\\
\hline\hline
\end{tabular}
%\vspace{-8mm}
\end{table}

\begin{table*}[!t]
\renewcommand{\arraystretch}{1.1}
\caption{Simulation results for the Hong Kong case}
\label{tab:HKsim}
\centering
\begin{tabular}{c|c || p{0.8cm}| c|c|c|p{1.2cm}|p{1.1cm}|p{1.1cm}||c|c|c|p{1.1cm}}
\hline\hline
\multirow{2}{*}{$\alpha$}	& \multirow{2}{*}{$D$ (km)} & \multirow{2}{0.8cm}{No. of Matched nodes}& \multicolumn{6}{|c||}{Objective function value (No. of stations constructed)}  & \multicolumn{4}{c}{Computation time (s)}\\ \cline{4-13} 
 &  &	& Method I &  Method II & Method III	& Method IV (mean) & Method IV (best) & Method IV (worst)	& Method I &  Method II & Method III	& Method IV (mean) \\
\hline
\multirow{5}{*}{1}	&30	& 4 & \textbf{19 181} (4)	& \textbf{19 181} (4)	& \textbf{19 181} (4)	& \textbf{19 181} (4)		& \textbf{19 181} (4)	& \textbf{19 181} (4)		 & 16.81	& 5.77E-03	& 3.77 	& 6.70\\
	&35	& 2 & \textbf{9571}	(2) & \textbf{9571}	(2)			& \textbf{9571}	(2) 	& \textbf{9571}	(2)			& \textbf{9571}	(2)	& \textbf{9571}	(2)				 & 17.14	& 5.46E-03	& 10.71	& 3.62\\
	&40	& 2 & \textbf{9571}	(2)	& \textbf{9571}	(2)			& \textbf{9571}	(2) 	& \textbf{9571}	(2)		& \textbf{9571}	(2)	& \textbf{9571}	(2)		&   17.31	& 5.86E-03	& 12.62	& 3.63\\
	&45 &	2 & \textbf{9527}	(2)	& \textbf{9527}	(2)			& \textbf{9527}	(2)	& \textbf{9527}	(2)			& \textbf{9527}	(2)	& \textbf{9527}	(2)				 &  17.89	& 5.75E-03	& 14.80	& 3.28\\
	&50 &	2 & \textbf{9527}	(2)		& \textbf{9527}	(2)			& \textbf{9527}	(2)		& \textbf{9527}	(2)			& \textbf{9527}	(2)		& \textbf{9527}	(2)			 & 17.85	& 5.56E-03	& 18.30	& 3.76\\
\hline
\multirow{5}{*}{0.8}	&30	& - & -	& -	& -	& -		& -	& - 		 & -	& -	& -	& -\\
	&35 &	1 & \textbf{23 768} (4)	& 30 001 (6)			& \textbf{23 768} (4) 	& \textbf{23 768} (4)		& \textbf{23 768} (4)	& \textbf{23 768} (4)		 & 17.17	& 6.88E-03	& 10.79	& 10.78\\
	&40	& 1 & \textbf{11 718} (2)	& 14 348 (4)			& \textbf{11 718} (2) 	& \textbf{11 718} (2)			& \textbf{11 718} (2)	& \textbf{11 718} (2) 				& 17.45	& 5.31E-03	& 12.69	& 4.59\\
	&45 & 2	& \textbf{9571} (2)	& \textbf{9571} (2)			& \textbf{9571} (2)	& \textbf{9571} (2)		& \textbf{9571} (2)	& \textbf{9571} (2) 		 & 17.76	& 5.30E-03	& 14.86	& 3.72\\
	&50 & 2	& \textbf{9571} (2)	& \textbf{9571} (2)			& \textbf{9571} (2)	& \textbf{9571} (2)		& \textbf{9571} (2)	& \textbf{9571} (2)		 & 17.86	& 5.59E-03	& 17.91	& 3.81\\
\hline
\multirow{5}{*}{0.6}	&30 & -	& -		& -		& -		& -			& -		& -			& - 		& - 		& -		& -	\\
	&35	& - & -		& -		& -		& -			& -		& -			& - 		& - 		& -		& -	\\
	&40	& - & -		& -		& -		& -			& -		& -			& - 		& - 		& -		& -	\\
	&45 &	4 & \textbf{27 889} (5)	& 30 065 (6)			&  \textbf{27 889} (5)	& 27 912.7 (5)	& \textbf{27 889} (5)		& 27 984 (5)		& 17.62	& 7.55E-03	& 14.73	& 9.35\\
	&50 & 4	& \textbf{19 181} (4)	& \textbf{19 181} (4)	& \textbf{19 181} (4)	& \textbf{19 181} (4)		& \textbf{19 181} (4)	& \textbf{19 181} (4)		& 17.79	& 4.83E-03	& 17.87	& 5.76\\
\hline\hline
\end{tabular}
\end{table*}

Methods I and III always find the best solutions for all cases. Method IV can determine the best solutions in most cases while Method II can still achieve the best solutions in some cases. For computation time, on the average, Method II is the fastest, and then Method IV. Method III comes next and Method I takes the longest. In general, the computation times of Methods I, II, and IV decrease with $D$ since $\hat{G}$ becomes denser with $D$ and we need less effort to locate the solutions. But Method III does the opposite because nodes tend to have larger degrees with $D$. Thus the number of MILPs needed to be solved increases together with the minimum degree of $\hat{G}$.

%----------------------------------------------------------------------
\subsection{Discussion}

From the simulations above, we can see that each method has its own characteristics and is suitable for different situations. Here we try to compare them in terms of five different perspectives independently:

\subsubsection{Solution quality}
If the adopted MILP solver can guarantee optimality\footnote{Most MIP solvers can generate the optimal solutions when the problem size is small.}, Methods I and III can obtain the best results. As Method II is embedded in Method IV,  Method IV is always superior to Method II but they may not produce the optimal solutions, especially when the solution space is getting larger. They can be ranked as: I$=$III$>$IV$>$II.

\subsubsection{Computation time}
Method II is of the simplest design and takes very limited amount of time to obtain a (usually sub-optimal) solution. Method I needs to apply MILP to all nodes of the problem while Method III requires only a subset. Hence, Method III is always  faster than Method I. Method IV is a metaheuristic and we can terminate the algorithm when certain stopping criteria are satisfied (in our cases, we limit the computation time by setting an FE limit). 
%\sout{We can allow the algorithm to run for different amounts of time to get solutions of different quality. Usually, the less the time, the worse the solutions produced. It is a tradeoff between solution quality and speed and it is hard to compare Method IV's efficiency with the others.} 
\textcolor{black}{As a whole, they can be ranked as: II$>$IV$>$III$>$I.} 

\subsubsection{Solvable problem size}
Solving MILP is the major building block of Methods I and III and it relies on the adopted solver. As most existing solvers can handle relatively small problems (in our case with MATLAB/YALMIP/CPLEX, the problem is only solvable with $n\leq 150$ in this study). However, since the manipulating mechanisms of Methods II and IV are mainly about how to modify and evaluate temporary solutions, they are more resistive to the growing problem size. So they can be ranked as: II$=$IV$>$I$=$III.

\subsubsection{Algorithmic nature}
All methods are deterministic except Method IV. In other words, we always come up with the identical result in different runs of the same problem instance. Method IV is probabilistic in nature. For each instance, we repeat the simulations several times to obtain its average performance.

\subsubsection{Prerequisite}
Recall that Method III is valid only when the condition imposed in Theorem \ref{thm:method3} holds. Although this condition is very general and held in most practical situations, the other methods do not require it.

We summarize the characteristics of the four methods in Table~\ref{tab:comparison}. Note that the above conclusions are drawn independently of each other from our observations on the simulation. When evaluating a particular method, we usually take several aspects into account simultaneously, e.g., correlating solution quality with computation time and problem size. However, we aim to give an extensive assessment and thus we appraise their individual abilities from one aspect to another. In general, each method has its own pros and cons and none is outstanding predominantly. In practice,  we select the most suitable method according to our need.

As a final remark, 
%Method IV has not been fully optimized. As discussed in \cite{CRO}, its performance is subject to the parameter settings and the duration of the simulation run. We can further improve its performance by determining a better combination of parameter values and allowing it to run for long time. However, 
we aim to show that the greedy algorithm can be considered as a component in a metaheuristic, whose performance is guaranteed to be better than that of Method II. 
%Seeking the best result for Method IV is out of the scope of this paper. Moreover, 
Method IV can be interpreted in a broader sense; it is a greedy-algorithm embedded metaheuristic approach where we can replace CRO with any metaheuristic like Genetic Algorithm \cite{GA}. Further exploration of Method IV with other metaheuristics will be left for the future work.  

\begin{table}[!t]
\renewcommand{\arraystretch}{1.1}
\caption{Solution method characteristic comparison}
\label{tab:comparison}
\centering
\begin{tabular}{p{2.1cm}|c | c|c|c}
\hline\hline
				& Method I 	& Method II & Method III & Method IV	\\
\hline
Solution quality & $\checkmark\checkmark\checkmark$		& $\checkmark$	& $\checkmark\checkmark\checkmark$ & $\checkmark\checkmark$\\
Computation time		& $\checkmark$	& $\checkmark\checkmark\checkmark\checkmark$ & $\checkmark\checkmark$ & $\checkmark\checkmark\checkmark$ \\
Problem size &	$\checkmark$ 		& $\checkmark\checkmark\checkmark$ 	& $\checkmark$	& $\checkmark\checkmark\checkmark$\\
algorithmic nature			 & Deterministic 		& Deterministic & Deterministic & Probabilistic \\
Prerequisite	 &	None & None		& Some & None\\
\hline\hline
\end{tabular}
\vspace{-5mm}
\end{table}

%----------------------------------------------------------------------
%----------------------------------------------------------------------
\section{Conclusion} \label{sec:conclusion}

Gasoline is a heavily demanded natural resource and most is consumed on transportation. Transportation electrification can relieve our dependence on gasoline and tremendously reduce the amount of harmful gases released, which partially constitute global warming and worsen our health. In the $21^\text{st}$ century, advancing EV technologies has become one of the keys to boost a nation's economy and maintain (and improve) people's quality of living. For long-term planning, modernizing our cities with EVs is of utmost importance. EVs will be integrated into the transportation system seamlessly and this will help make our cities ``smart''. To do this, we need to determine the best locations to construct charging stations in the city. In this paper, we focus on human factors rather than technological ones for charging station placement. An EV should always be able to access a charging station within its driving capacity anywhere in the city. 
Our contributions in this paper include: 1) formulating the problem, 2) identifying its properties, and 3) developing the corresponding solution methods. We formulate the problem as an optimization model, based on the charging station coverage and the convenience of drivers. We prove the problem NP-hard and propose four solution methods to tackle the problem. Each method has its own characteristics and is suitable for different situations depending on the requirements for solution quality, algorithmic efficiency, problem size, nature of the algorithm, and existence of system prerequisite.

% if have a single appendix:
\appendix[An illustrative example for the network flow model]
\textcolor{black}{
We make use of  the example of $\hat{G}$ given in Fig. \ref{fig:illstrative} to illustrate the network flow model discussed in Section \ref{subsec:formualtion}.
}
\begin{figure}
\centering
\includegraphics[width=3.0in]{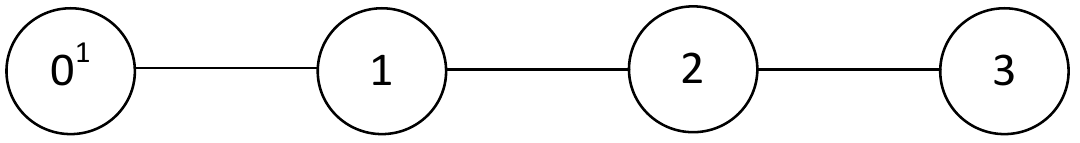}
 \caption{An example of four nodes.}
\label{fig:illstrative}
\end{figure}

\textcolor{black}{
Consider that nodes 1 and 2 have charging stations constructed and thus we have node $0^1$ sends out two units of flow on $(0^1,1)$ and hence $y_{01}^1=2$. Eq. \eqref{subcon3} indicates that the conversation of flow and thus we have $y_{12}^1=1$ as node 1 is a sink of one unit of flow. Similarly, we get $y_{23}^1=0$. In this way, we can ensure that the resultant locations of the charging stations (nodes 1 and 2 in this case) are connected.
}

\textcolor{black}{
Consider another case that nodes 1 and 3 are the locations of charging stations. So we have $x_1=x_3=1$ and $x_2=0$. Eq. \eqref{subcon4} results in $y_{01}^1=2$. Node 2 is not a sink and \eqref{subcon2} confines  $y_{12}^1=0$. To balance the incoming and outgoing flows at node 2, \eqref{subcon3} makes $y_{23}^1=0$.
 However, node 3 is a sink of one unit and when \eqref{subcon3} is applied to node 3, we need to have $y_{23}^1=1$. This results in a contradiction and hence we cannot allow constructing charging stations at nodes 1 and 3 without node 2.
}

\textcolor{black}{
Therefore, connectivity of the charging station network can be enforced with the network flow model.
}
%\appendices
%\section{Proof of the First Zonklar Equation}
%Appendix one text goes here.
%
%% you can choose not to have a title for an appendix
%% if you want by leaving the argument blank
%\section{}
%Appendix two text goes here.
%
%
%% use section* for acknowledgement
%\section*{Acknowledgment}
%
%
%The authors would like to thank...

% Can use something like this to put references on a page
% by themselves when using endfloat and the captionsoff option.
\ifCLASSOPTIONcaptionsoff
  \newpage
\fi

% trigger a \newpage just before the given reference
% number - used to balance the columns on the last page
% adjust value as needed - may need to be readjusted if
% the document is modified later
%\IEEEtriggeratref{8}
% The "triggered" command can be changed if desired:
%\IEEEtriggercmd{\enlargethispage{-5in}}

% references section
\bibliographystyle{IEEEtran}
% argument is your BibTeX string definitions and bibliography database(s)
\bibliography{IEEEabrv}
%

%\begin{thebibliography}{1}
%
%\bibitem{IEEEhowto:kopka}
%H.~Kopka and P.~W. Daly, \emph{A Guide to \LaTeX}, 3rd~ed.\hskip 1em plus
  %0.5em minus 0.4em\relax Harlow, England: Addison-Wesley, 1999.
%
%\end{thebibliography}

%\begin{IEEEbiography}{Michael Shell}
%Biography text here.
%\end{IEEEbiography}

\end{document}